\documentclass[letterpaper, 10 pt, conference]{ieeeconf}  

\IEEEoverridecommandlockouts                              
\usepackage[utf8]{inputenc}
\usepackage[T1]{fontenc}
\let\labelindent\relax

\usepackage{amsmath}
\usepackage{amssymb}
\usepackage{algorithmic}
\usepackage{algorithm}
\usepackage{mathtools}
\usepackage[caption=false, font=footnotesize]{subfig}
\usepackage{xcolor}
\usepackage{makecell}
\usepackage{multirow}
\usepackage{graphicx}
\usepackage{marginnote}
\usepackage{float}
\usepackage{gensymb}
\usepackage{amsthm}
\newtheorem{remark}{Remark}
\newtheorem{assumption}{Assumption}
\newtheorem{thm}{Theorem}

\newcommand{\nc}{\mathrm}

\newcommand{\n}{\mathbf}
\newcommand{\cm}{\mathcal}

\usepackage[bottom]{footmisc}
\usepackage{enumitem}
\usepackage[noadjust]{cite}

\usepackage[scr=boondox]{mathalfa}

\title{\LARGE \bf
Cooperative nonlinear distributed model predictive control\\ with dissimilar control horizons}

 \author{Paula Chanfreut, José M. Maestre, Quanyan Zhu, W.P.M.H. (Maurice) Heemels
\thanks{*This work is supported by the
ERC Advanced Grant OCONTSOLAR (SI-1838/24/2018), and by the {Spanish} projects C3PO-R2D2 and C3PO-R3 under grants PID2020-119476RB-I00 and {PID2023-152876OB-I00.}}
\thanks{P. Chanfreut and M. Heemels are with the Department of Mechanical Engineering, Eindhoven University of Technology, The Netherlands. E-mails:
        {\tt\small p.chanfreut.palacio@tue.nl, m.heemels@tue.nl}.}%
\thanks{J. M. Maestre is with the Department of Systems and Automation Engineering, University of Seville, Spain. E-mail:
        {\tt\small pepemaestre@us.es}.}
        \thanks{Q. Zhu is with the Department of Electrical and Computer Engineering, New York University, USA. E-mail:
        {\tt\small qz494@nyu.edu}.}%
}

\begin{document}
\maketitle


\begin{abstract}
In this paper, we introduce a nonlinear distributed model predictive control (DMPC) algorithm, which allows for dissimilar and time-varying control horizons among agents, thereby addressing a common limitation in current DMPC schemes. We consider  cooperative agents with varying  computational capabilities and operational objectives, each willing to manage varying numbers of optimization variables at each time step. Recursive feasibility and a non-increasing evolution of the optimal cost are proven for the proposed  algorithm. Through numerical simulations on systems with three agents, we show that our approach effectively approximates the performance of traditional DMPC, while reducing the number of variables to be optimized. This advancement paves the way for a more decentralized yet coordinated control strategy in various applications, including power systems and traffic~management.
\end{abstract}


%

\section{Introduction}

Distributed model predictive control (DMPC) has gained considerable attention to address the inherent challenges posed by the control of large-scale systems and overcome the limitations of centralized and decentralized model predictive control (MPC) architectures~\cite{negenborn2014distributed}. Centralized MPC employs a single system-wide controller to achieve optimal global performance, but presents significant scalability issues and lacks redundancy. In contrast, decentralized MPC spreads decision-making across multiple agents, each managing a subsystem. The decentralized performance, however, degrades as the coupling between subsystems becomes stronger, as it lacks inter-agent coordination~\cite{rawlings2008coordinating}. DMPC adopts a similar agent-based decomposition and solves this issue at the expense of increased communication and computational~complexity.

  
A closer look at DMPC shows that each agent typically manages a subset of control inputs, while recurrently sharing information with its peers to achieve a certain level of global coordination~\cite{scattolini2009architectures}. As reviewed by~\cite{negenborn2014distributed}, a wide range of DMPC approaches can be found in the literature with significant differences on their basic assumptions, e.g., on the type of coordination  mechanism employed, the source of coupling between agents (dynamics, objectives, or constraints), their attitude (cooperative or non-cooperative), and the need for performing iterative computations. In this regard, some comparative works have pointed out a clear trade-off between communication burden and performance~\cite{alvarado2011comparative,maestre2015comparison}, which has led to the development of clustering-based MPC strategies where agents can merge and switch between decentralized and distributed MPC strategies to save coordination efforts with minimal effects on performance~\cite{CHANFREUT2021_survey}.

In this work, we deepen into the previously mentioned trade-off in a different manner. In particular, a common assumption across DMPC schemes is that all agents use identical control horizons, which is somehow limiting for heterogeneous systems~\cite{bao2022recent}. For example, it is not difficult to imagine agents with different computational capacities, hence being able to handle different numbers of optimization variables and problems complexities.  Likewise, they may possess varying levels of reactivity and proactivity in their decisions, which may not align with those of their peers. Moreover, the agents dynamics and objectives themselves may require different prediction horizons, e.g., in traffic freeways and irrigation canals agents manage segments of variable lengths with the corresponding transport delays~\cite{fele2014coalitional, chanfreut2020coalitional}. In all these situations, it is desirable to have a DMPC scheme that can accommodate unequal, and potentially time-varying, horizons according to their computational capacities and yet enjoy the benefits of coordinated control actions. 

Based on the above, we propose a novel nonlinear DMPC algorithm inspired in~\cite{stewart2011cooperative}, which is in turn based on the linear DMPC proposed by~\cite{venkat2004plant}. Numerical results on a system with three agents are presented, showcasing certain combinations of unequal control horizons can closely approximate the performance achieved with equal longer horizons. The latter comes with the added benefits of reduced computational and communication overheads, thereby also approaching the decentralized simplicity. Note that the proposed DMPC  partially shares the  underlying goal of~\cite{ma2020event} and~\cite{wang2023event}, which present  event-triggered linear DMPC methods with variable prediction horizons  to reduce the complexity of the optimization problem. We contribute to this line of research focusing on nonlinear cooperative agents with a different MPC problem formulation.


The remainder of this paper is organized as follows. Section~\ref{sec:system_description} formulates the problem. Section~\ref{sec:varying} presents the DMPC formulation with dissimilar and time-varying control horizons. Section~\ref{sec:simulations} includes our simulation results. Finally, concluding remarks are given in Section~\ref{sec:conclusions}.

\textit{Notation:} The set of non-negative natural numbers is $\mathbb{N}=\{0,1,2, \hdots\}$, whereas $\mathbb{N}_{\geq a}=\{a,a+1,a+2, \hdots\}$ for any ${a\in \mathbb{N}}$. Also, given two scalars $a,b\in \mathbb N$ with~${b>a}$, we define set $\mathbb{N}_{[a,b]}=\{a, a+1, ..., b\}$. In addition, $[v_i]_{i\in \mathbb{N}_{[a,b]}}$ denotes the column vector $[v_a^\top, v_{a+1}^\top, ...., v_{b}^\top]^\top$ for any variable $v$, and  the notation $v(t|k)$ indicates the predicted value of~$v$ made at instant~$k$ for time instant $k+t$.   Finally,~$\otimes$ denotes the Cartesian product.

\section{Problem setting}\label{sec:system_description}

This section introduces the system dynamics, presents the distributed control architecture, and describes the goals considered throughout this article. 

 \subsection{System description}
Consider a class of systems that can be partitioned into a set $\cm{N}=\{1,2,...,N\}$ of coupled subsystems with dynamics
\begin{equation}
\begin{split}
&x_{i}(k+1) = f_{i}(x(k), u(k)), \\
    \end{split}
    \label{eq:subsystems_model}
\end{equation}
\noindent where $x_i(k)\in \mathbb R^{n_i}$ and $u_i(k)\in \mathbb R^{m_i}$ denote, respectively, the state and input of subsystem $i\in \cm{N}$ at time instant~$k\in\mathbb N$, and $x(k)=[x_i(k)]_{i\in\cm{N}}$ and $u(k)=[u_i(k)]_{i\in\cm{N}}$ are the global state and input vectors. In addition,~$f_{i}: \mathbb R^{n} \times \mathbb R^{m} \rightarrow \mathbb R^{n_i}$ is a possibly nonlinear function for all~${i\in\cm{N}}$, with $n=\sum_{i\in\cm{N}} n_i$ and $m=\sum_{i\in \cm{N}} m_i$. Note also that, considering~\eqref{eq:subsystems_model} for all~$i\in \cm{N}$, the global system dynamics can be modeled~as
    \begin{equation}
    x(k+1)=f(x(k), u(k)),
   \label{eq:global_model}
\end{equation}
 where $f$ aggregates functions $f_{i}$ for all~{$i\in \cm{N}$}, i.e., $f=(f_1, f_2, \hdots, f_N)$.




As for the constraints, we require the state and input of every subsystem $i\in \cm{N}$ to satisfy
\begin{equation}\label{eq:constraints_sets}
    x_i(k) \in \mathcal{X}_i, \ \   u_i(k) \in \mathcal{U}_i, \ \ k\in \mathbb N,
\end{equation}
where $\mathcal{X}_i$ and $\mathcal{U}_i$ are compact convex sets that contain the origin in their interior. 

Finally, let us define the stage performance cost of each subsystem $i\in \cm{N}$ at time $k\in \mathbb N$ as~$\ell_i(x_i(k), u_i(k))$. Similarly, the global stage cost will be given by  $\ell(x(k), u(k)) = \sum_{i\in\cm{N}} \ell_i(x_i(k), u_i(k))$.

\subsection{Multi-agent control architecture}
In what follows, consider that the set of subsystems in $\cm{N}$ is managed by a set of \emph{cooperative} MPC agents. Specifically, each subsystem $i \in \cm{N}$ is assigned to local MPC agent $i$, which determines control input~$u_i(k)$ at every time instant~${k\in\mathbb N}$. 

The local decisions will be cooperatively negotiated to optimize the global performance following an algorithm inspired by~\cite{stewart2011cooperative}. To this end, the set of agents are allowed to share data through a communication network  as illustrated in Fig.~\ref{fig:distributed_system}. For simplicity, let us assume that the topology of this network is modeled by fully connected graph $\cm{G}=(\cm{N},\cm{E})$, where $\cm{E}=\{(i,j) : i, j \in \cm{N}, i\neq j\}$, and let us introduce the following assumption:
\vspace{4pt}
\begin{assumption}
    At every step $k$, all agents know the global system state $x(k)$.
\end{assumption}
\noindent Finally, notice that the assumption above is also considered in~\cite{stewart2011cooperative}, and can be relaxed by following~\cite{razzanelli2017parsimonious}, allowing agents to operate with partial knowledge of the system state. However, the latter has been left out of the scope of this paper and will be considered in future works.

\begin{figure}[t]
    \centering
    \includegraphics[scale=0.55,trim={5.9cm 6.7cm 8.5cm 6.5cm},clip]{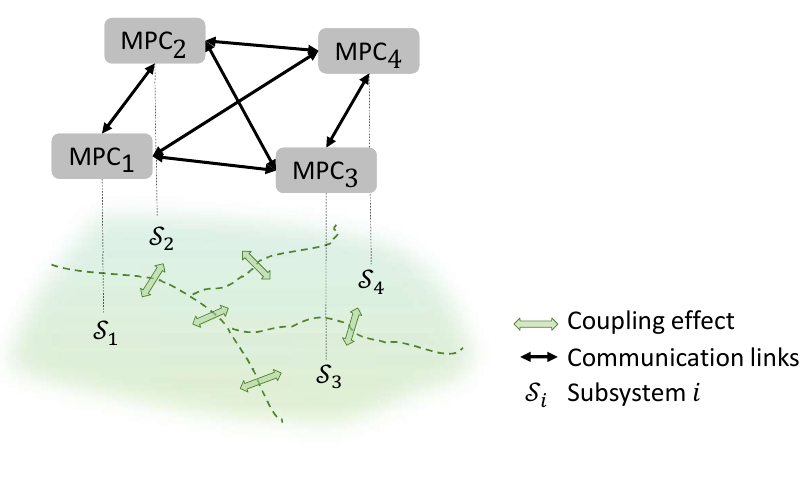}
    \vspace{-0.75cm}
    \caption{Distributed system with 4 subsystems managed by MPC agents.  }
    \label{fig:distributed_system}
\end{figure}

\subsection{Main goal}

The underlying goal of this article is to develop a cooperative nonlinear DMPC, where the agents  manage  \textit{efficiently} their computation resources while approximately minimizing the following global cost:
\begin{equation}\label{eq:cost_infinity}
       \sum\limits_{k=0}^{\infty}  \sum_{i \in \cm{N}} \ell_i(x_i(k), u_i(k)). 
\end{equation}
\textit{Efficiency} will be measured considering the time required by the agents to solve their MPC problems and the resulting global performance. We will exploit the role of the control horizons, allowing them to be adjusted unequally and dynamically. Note that the control horizon determines the number of optimization variables handled by each of the~agents.

\section{Nonlinear DMPC with dissimilar control horizons}\label{sec:varying}

This section presents the proposed DMPC and its properties. For its use as a reference, let us first introduce the following \emph{centralized} MPC problem for state $x(k)$ and time~$k\in\mathbb N$:
\begin{subequations}\label{eq:Cent_MPC}
\begin{align}
 &\hspace{-25pt} \min_{[\mathbf{u}_i(k)]_{i\in\cm{N}}} \ \   J(x(k), [\mathbf{u}_i(k)]_{i\in\cm{N}}) \nonumber\\
\nc{s.t.} \ 
   &x(0|k) = x(k),\label{eq:ini_state} \\[3pt]
\begin{split}
    & x(t+1|k) = f(x(t|k), u(t|k)),  \ \forall t\in \mathbb{N}_{[0,N_\nc{p}-1]}, 
\end{split} \label{eq:pred_model}\\[3pt]
&   x_i(t|k) \in \mathcal{X}_i, \ \  \forall t\in \mathbb{N}_{[0,N_\nc{p}-1]}, \ \label{eq:state_constraints} \\[3pt]
&   u_i(t|k) \in \mathcal{U}_i, \ \  \forall t\in \mathbb{N}_{[0,N_\nc{p}-1]}, \  \\[3pt]
&  x(N_\nc{p}|k) \in \mathcal{X}_\nc{f}, \label{eq:term_constraint_global} \\[3pt]
& \forall i \in \cm{N}. \nonumber
\end{align}
\end{subequations}
\noindent Here, the objective function is given by 
\vspace{4pt}
\begin{equation}
    \begin{split} J&(x(k), [\mathbf{u}_i(k)]_{i\in\cm{N}}) =  \\ &\sum_{i \in \cm{N}}   \sum\limits_{t=0}^{N_\nc{p}-1}  \ell_i(x_i(t|k), u_i(t|k))   + V_\nc{f}(x(N_\nc{p}|k)),\end{split}
\end{equation}
where  $\n{u}_i(k)= [u_i^\top(0|k), \ u_i^\top(1|k),  \ \hdots, \ u_i^\top(N_\nc{p}\!-\!1|k)]^\top$ 
for all $i\in \cm{N}$, and $N_\nc{p}$ is the prediction horizon. 
Also, $\mathcal{X}_\nc{f}$ and $V_\nc{f}:\mathbb R^n \rightarrow \mathbb R_{\geq 0}$ represent respectively a  terminal set and terminal cost designed considering the following assumption:  

\vspace{4pt}
\begin{assumption}\label{as:terminal_components}
There exist a state-dependent control law $\kappa(x)=[\kappa_{i}(x)]_{i\in\cm{N}}$, a function $V_\nc{f}:\mathbb R^n \rightarrow \mathbb R_{\geq 0}$, and a set $\cm{X}_\nc{f} = \{x \in \mathbb R^n \ : \ V_\nc{f}(x) \leq \alpha, \ x\in \cm{X}, \ Kx \in \cm{U}\}$ such that for all $x  \in \cm{X}_\nc{f}$,
\begin{equation}\label{eq:cond_term_comp}
\begin{split}
    &V_\nc{f}(f(x, \kappa(x))) + \ell(x,\kappa(x)) - V_\nc{f}(x) \leq 0, 
\end{split}
\end{equation}
where $\alpha \in \mathbb R$, $\cm{X} = \otimes_{i\in\cm{N}} \cm{X}_i$, and $\cm{U} = \otimes_{i\in\cm{N}} \cm{U}_i$.
 \end{assumption}

\subsection{Control algorithm }

This article explores a nonlinear DMPC inspired by the cooperation-based approach presented in~\cite{venkat2008distributed,stewart2011cooperative}. In this setting, all agents optimize through an iterative procedure plant-wide performance function $J(\cdot)$, thus taking into account the global effect of the local decisions. To enumerate the iterations, let us use in what follows subscript $p \in \mathbb N$, i.e.,~$\mathbf{u}_i^p(k)$ represents the value of $\mathbf{u}_i(k)$ at iteration $p$.

The pseudocode of the distributed strategy proposed in this article is given in Algorithm~\ref{alg_1}. As can be seen, at every time step, the agents perform an iterative negotiation, where every iteration comprises the following steps:\footnote{For the sake of clarity, we omit time step~$k$ in the rest of this subsection when it can be clearly inferred from the context.}
\begin{itemize}
    \item[(i)] All agents $i \in \cm{N}$ solve (in parallel) nonlinear MPC problems~\eqref{eq:Dist_MPC}. By doing so, they obtain local input sequence $\n{u}_i^\ast$, together with scaling factors~$\lambda_{i}^\ast$, which are introduced to guarantee recursive feasibility.  Notice that  only the local inputs associated with agent~$i$ are optimization variables in~\eqref{eq:Dist_MPC}, whereas those of all~$j\neq i$ are fixed to their values at iteration $p-1$.
    \item[(ii)] The agents solutions are optimally combined (if necessary) by a supervisory entity using weights $[\gamma_i]_{i\in\cm{N}}$ as will be detailed below (see~\eqref{eq:coord_problem}). 
\end{itemize}

\begin{algorithm}[t]
\caption{}\label{alg_1}
Let $k$ be the current time instant, and define an initial feasible input  sequence~$\n{u}_i^0$ (see Remark~\ref{rem:initial_u}), and a control horizon~$N_{\nc{c}, i}$, for all $i\in \cm{N}$.   Then, starting from $p=1$, proceed as follows:
\begin{algorithmic}[1]
\STATE All MPC agents $i \in \cm{N}$ compute (in parallel) $\mathbf{u}_i^{\ast}$ by solving:
\begin{subequations}\label{eq:Dist_MPC}
\begin{align}
[\mathbf{u}_i^\ast, \ &\lambda_i^\ast] =   \arg  \min_{\mathbf{u}_i, \lambda_i} \  J(x(k),  \mathbf{u}_i, [\mathbf{u}_j]_{j\in\cm{N}_{-i}}) \qquad \qquad \ \ \nonumber\\
\nc{s.t.} \  
&  \eqref{eq:ini_state} \text{ to } \eqref{eq:term_constraint_global}, \nonumber \\[3pt]
   \begin{split}
    &u_i(t|k) = \lambda_{i}\kappa_i(x(t|k)) + (1\!-\!\lambda_{i}) u_i^{p-1}(t|k), \\   &\forall t\in \mathbb{N}_{[N_{\nc{c},i},N_\nc{p}-1]},
\end{split} \label{eq:ui_from_Nci}\\[3pt] 
& \mathbf{u}_j =  \mathbf{u}_j^{p-1}, \ \forall j \in \cm{N}_{-i}, \\[3pt]
& \lambda_{i} \in [0,1],
\end{align}
\end{subequations}
\noindent where $\cm{N}_{-i} = \cm{N} \setminus \{i\}$. As for the objective function above, note that we can rewrite $
    J(x(k), [\mathbf{u}_i]_{i\in\cm{N}}) = J(x(k), \n{u}_i, [\n{u}_j]_{j \in \cm{N}_{-i}})$.
\STATE All MPC agents $i \in \cm{N}$ share $\mathbf{u}_i^{\ast}$ with the supervisor.
\STATE The supervisor defines weights $[\gamma_{i}^\ast]_{i\in\cm{N}}$ as follows:
\begin{subequations}\label{eq:coord_problem}
\begin{align}
[\gamma_{i}^\ast]_{i\in\cm{N}}& =  \arg \min_{[\gamma_{i}]_{i\in\cm{N}}} \  J(x(k),  [\mathbf{u}_i]_{i\in\cm{N}}) \nonumber \qquad \qquad \qquad  \\[4pt]
\nc{s.t.} \ &\eqref{eq:ini_state} \text{ to } \eqref{eq:term_constraint_global}, \nonumber \\
&   \n{u}_i = \gamma_{i}\mathbf{u}_i^{\ast} + (1- \gamma_{i})\n{u}_{i}^{p-1}, \\[3pt] 
&  \gamma_{i} \in [0,1], \\[3pt] 
& \forall i \in \cm{N}. 
\end{align}
\end{subequations}
\STATE All MPC agents $i \in \cm{N}$ define their input sequence for iteration $p$ as $
    \n{u}_i^{p} = \gamma_{i}^\ast \mathbf{u}_i^{\ast} + (1-\gamma_{i}^\ast) \n{u}_{i}^{p-1}$, and share it with the other agents.
\STATE Set $p=p+1$ and go back to Step 1 until convergence (or a maximum number of iterations) is reached.
\STATE All agents implement inputs $u_i^{\bar{p}}(0|k)$, where $\bar{p}$ denotes the last iteration index. 
\end{algorithmic}
\end{algorithm}

 Regarding (i),  note that the control horizons used to solve problem~\eqref{eq:Dist_MPC} can differ among the set of agents. In particular, for any agent $i\in \cm{N}$, its control horizon is denoted as~$N_{\nc{c},i}$ and is such that $N_{\nc{c},i} \leq N_\nc{p}$.  From instant~$N_{\nc{c},i}$ of the prediction, the local inputs are defined as a linear combination of the previously obtained solution and a given predefined control law (see~\eqref{eq:ui_from_Nci}). For simplicity, we have considered control law $\kappa_i$ (recall~Assumption~\ref{as:terminal_components}), but notice that a different approach can also be used, e.g., the inputs may be kept constant from instant $N_{\nc{c},i}$. 
 

 Regarding (ii), let us remark that the resulting combination of solutions $\n{u}_i^\ast$ for all $i\in \cm{N}$ may not be globally feasible. To illustrate this, consider problem~\eqref{eq:Dist_MPC} with a simplified scenario involving only two agents. In this context, we have that $(\n{u}_1^\ast, \n{u}_2^{p-1})$ and $(\n{u}_1^{p-1}, \n{u}_2^\ast)$ constitute solutions satisfying the constraints of~\eqref{eq:Cent_MPC} at every iteration~$p$. However, the same cannot be guaranteed for the solution~$(\n{u}_1^\ast, \n{u}_2^\ast)$. This is the reason motivating the introduction of  weights~$[\gamma_i]_{i\in\cm{N}}$, which allow steering the inputs towards those of iteration $p-1$ when it is globally advantageous.  

In addition, let us mention that, although Algorithm~\ref{alg_1} is inspired by~\cite{stewart2011cooperative}, there are significant differences to highlight. Specifically, both the formulation of problem~\eqref{eq:Dist_MPC} and the final definition of sequences $[\n{u}_i^p]_{i\in\cm{N}}$ at every iteration $p$ differ from those in~\cite{stewart2011cooperative}. The main differences are as follows. Firstly, the agents not only consider input constraints but also account for state constraints (see~\eqref{eq:state_constraints} and~\eqref{eq:term_constraint_global}). Secondly, we introduce constraint~\eqref{eq:ui_from_Nci}, which is formulated not to compromise recursive feasibility while facilitating the reduction of the number of optimization variables. Thirdly, the set of agents may work with dissimilar (and possibly dynamic) horizons~$N_{\nc{c},i}$. Lastly, the final sequences $[\n{u}_i^p]_{i\in\cm{N}}$ are determined as an optimal combination of new and previous solutions by solving~\eqref{eq:coord_problem}. 



\begin{remark}\label{rem:initial_u}
    By construction of Algorithm~\ref{alg_1}, input sequences~$\n{u}_i^{p}(k)$  for $i\in \cm{N}$ constitute a feasible solution of problem~\eqref{eq:Cent_MPC} at any time instant $k \in \mathbb N$ and iteration~$p\in \mathbb N$. Therefore, given the definition of terminal set~$\cm{X}_\nc{f}$ (recall Assumption~\ref{as:terminal_components}), we can build a candidate solution for any agent $i$ at instant~${k+1}$~as
    \begin{equation}\label{eq:init_u}
       \tilde{\mathbf{u}}_i(k+1)=  \begin{bmatrix}
           [u_i^{\bar{p}(k)}(t|k)]_{t=1}^{N_\nc{p}-1} \\[3pt] \kappa_i(x^{\bar{p}(k)}(N_\nc{p}|k))
        \end{bmatrix} .
    \end{equation}
    This will be used to define $\mathbf{u}_i^0(k+1)$ at every time instant~$k$.
\end{remark}

  \begin{remark}
  Problem~\eqref{eq:coord_problem} and step 4 of Algorithm~\ref{alg_1} could be omitted without compromising its essence and the properties indicated in the next subsection. In particular,  note that, after solving~\eqref{eq:Dist_MPC}, each agent $i$ \emph{proposes} global solution $(\n{u}_i^\ast, [\n{u}_j^{p-1}]_{j\in \cm{N}_{-i}})$, which is feasible by definition of the constraints in~\eqref{eq:Dist_MPC}.  Following the game-theoretic notion of dominant-strategies, the solution providing greater global benefits may also be chosen. That is, the proposal \emph{dominating} the rest could be selected as the input sequences of iteration~$p$.  Nevertheless, solving~\eqref{eq:coord_problem} allows for greater reductions of the cost function along the iterations. 
   \end{remark}


\vspace{6pt}
\subsection{Properties}

This subsection proves recursive feasibility and a non-increasing evolution of the cost function for the proposed DMPC.  In the following theorems, we define sequences~$[\n{u}_i^\cdot(k)]_{i\in\cm{N}}$ as \emph{globally feasible} at instant $k$,  if they satisfy the constraints of~\eqref{eq:Cent_MPC}. That is, the inputs and resulting states belong in the corresponding constraints sets (see~\eqref{eq:constraints_sets}), and the terminal state is in the terminal set. 
\vspace{4pt}
\begin{thm}[Recursive feasibility]
    If the initial sequences $[\n{u}_i^0(0)]_{i\in\cm{N}}$ are globally feasible, then problems~\eqref{eq:Dist_MPC} for all agents $i\in\cm{N}$ and problem~\eqref{eq:coord_problem} will be recursively feasible. 
\end{thm}
\begin{proof}
    To prove this theorem, let us first consider any time instant $k\in \mathbb N$ and two consecutive iterations, say $p-1$ and~$p$. Then, 
    note that if sequences $[\n{u}_i^{p-1}(k)]_{i\in\cm{N}}$ are globally feasible, the following holds at iteration $p$:
    \begin{itemize}
        \item[i.] For all~$i\in \cm{N}$,  $\n{u}_i^{p-1}(k)$ provides a feasible solution of~\eqref{eq:Dist_MPC}. That is, the agents can always keep the  input sequence of the previous iteration.  Particularly, note that if we set $\lambda_{i}=0$, then $u_i(t|k)=u_i^{p-1}(t|k)$ for all~${t\in [N_{\nc{c},i}, N_{\nc{p}}-1]}$. 
        \item[ii.] By construction of problem~\eqref{eq:coord_problem},  a feasible solution can be built simply by setting $\gamma_i= 0$ for all $i\in\cm{N}$. Note that this translates into~${\n{u}_i^p(k) = \n{u}_i^{p-1}(k)}$. 
    \end{itemize}  
    \noindent As a consequence, if~$\n{u}_i^0(k)$ is globally feasible, $\n{u}_i^{1}(k)$ will also be globally feasible. By induction, we have that for any $p\in \mathbb N_{\geq 1}$ both problems~\eqref{eq:Dist_MPC} and~\eqref{eq:coord_problem} will have feasible solutions at time instant $k$. 

    We are left to prove that, at any $k\in \mathbb N_{\geq 1}$, it is possible to find  globally feasible input sequences for iteration 0. To prove this,  consider~\eqref{eq:init_u} and note that the latter can be obtained by setting $\n{u}_i^{0}(k+1)=\tilde{\mathbf{u}}_i(k+1)$ for all $i\in \cm{N}$. This concludes the proof. 
\end{proof}

\begin{thm}[Non-increasing objective function]
    The value of cost function $J(x(k),[\n{u}_i^p(k)]_{i\in\cm{N}})$ is non-increasing with respect to both time $k$ and iterations index~$p$. 
\end{thm}
\begin{proof}
    Consider an agent $i\in\cm{N}$ and note that, by optimality of $[\n{u}_{i}^{\ast}(k)]_{i\in\cm{N}}$ and construction of problem~\eqref{eq:coord_problem}, we have that
    \begin{equation}
        J(x(k),[\n{u}_{i}^{p}(k)]_{i\in\cm{N}}) \leq J(x(k),[\n{u}_{i}^{p-1}(k)]_{i\in\cm{N}}).
    \end{equation}
    In particular, by setting $\gamma_i=0$ for all $i \in \cm{N}$, the previous inequality holds as an equality. 
    In addition, considering initialization~\eqref{eq:init_u}, together with the properties of the terminal components, we have that:
    \begin{equation*}
       J(x(k+1), [\n{u}_{i}^{0}(k+1)]_{i\in\cm{N}}) \leq  J(x(k),[\n{u}_{i}^{\bar{p}(k)}(k)]_{i\in\cm{N}}).
    \end{equation*}
    Therefore, 
     \begin{equation}
     \begin{split}
                J(x(k+1), &[\n{u}_{i}^{\bar{p}(k+1)}(k+1)]_{i\in\cm{N}}) \\ &\leq        J(x(k+1), [\n{u}_{i}^{0}(k+1)]_{i\in\cm{N}}) \\
                &\leq  J(x(k),[\n{u}_{i}^{\bar{p}(k)}(k)]_{i\in\cm{N}}).
     \end{split}
    \end{equation}
    That is, the value of the global objective function after the iterative negotiation does not increase over time. 
\end{proof}

Note that the properties above depend on the existence of input sequences satisfying the constraints of \eqref{eq:Cent_MPC} only at initial time instant 0. Besides, these initial sequences do not have to consider any predefined control law as later introduced in~\eqref{eq:Dist_MPC}. Finally, notice that these proofs are not straightforwardly derived from the results in~\cite{stewart2011cooperative} due to the differences introduced in the proposed Algorithm~\ref{alg_1}.

\vspace{10pt}



\subsection{Extension to dynamic control horizons}\label{sec:varying_Nc}
This subsection introduces a heuristic approach to allow the agents to reduce their control horizons in~\eqref{eq:Dist_MPC} over time, and thus their number of optimization variables, in case they are not providing significant performance benefits. Note that the solution of~\eqref{eq:Dist_MPC} at any iteration has the following~form:
\begin{equation*}
    \n{u}_i^\ast= \begin{bmatrix}
    u_i^\ast(0|k) \\  \vdots \\   u_i^\ast(N_{\nc{c},i}-1|k) \\
        \lambda_{i}^\ast \kappa_i(x(N_{\nc{c},i}|k)) + (1-\lambda_{i}^\ast) u_i^{p-1}(N_{\nc{c},i}|k) \\ \vdots \\
        \lambda_{i}^\ast \kappa_i(x(N_\nc{p}-1|k)) + (1-\lambda_{i}^\ast) u_i^{p-1}(N_{\nc{p}}-1|k)
        \end{bmatrix}\!.
\end{equation*}
Considering this, we can build a slightly different input sequence, say $\hat{\n{u}}_i$, where control law $\kappa_i(\cdot)$ starts being considered one instant before, that is, 
\begin{equation*}
    \hat{\n{u}}_i= \begin{bmatrix}
    u_i^\ast(0|k) \\  \vdots \\   u_i^\ast(N_{\nc{c},i}-2|k) \\  
        \lambda_{i}^\ast \kappa_i(x(N_{\nc{c},i}-1|k)) + (1-\lambda_{i}^\ast) u_i^{p-1}(N_{\nc{c},i}-1|k) \\ \vdots \\
        \lambda_{i}^\ast \kappa_i(x(N_\nc{p}-1|k)) + (1-\lambda_{i}^\ast) u_i^{p-1}(N_{\nc{p}}-1|k)
        \end{bmatrix}\!.
\end{equation*}
If $\hat{\n{u}}_i$ is feasible, then it provides a upper bound on the objective function value of~\eqref{eq:Dist_MPC} when using control horizon~${N_{\nc{c},i}-1}$. Using this new candidate sequence~$\hat{\n{u}}_i$, we establish that if
\begin{equation*}
   J(x(k), \hat{\n{u}}_i, [\n{u}_j^{p-1}]_{j \in \cm{N}_{-i}}) -  J(x(k), \n{u}_i^\ast, [\n{u}_j^{p-1}]_{j \in \cm{N}_{-i}})   \leq \epsilon,
\end{equation*}
agent $i$ can reduce its control horizon by one unit, i.e., $N_{\nc{c},i} \leftarrow N_{\nc{c},i}-1$. Above, $\epsilon \in \mathbb R_{\geq 0}$ is a tuning parameter that can be chosen arbitrarily close to~0.  This heuristic approach will be implemented along the iterations  starting from a given horizon, say $N_{\nc{c},i}^0$. Note that, given the introduction of scalars~$\lambda_i$ and the definition of problem~\eqref{eq:coord_problem}, the properties introduced in the previous subsection are not compromised by dynamic changes of the control horizons. 

Finally, let us also remark that, instead of building a candidate input sequence, learning methods are also useful to optimize when and how to update~$N_{\nc{c},i}$ for all agents $i\in \cm{N}$. 

\section{Simulation results}\label{sec:simulations}

In this section, we apply the DMPC scheme with varying control horizons to a system with three masses connected by springs and dampers~\cite{liu2014distributed, li2013robust}. The system dynamics are given by 
\begin{equation}\label{eq:system_simulations} \resizebox{0.97\columnwidth}{!}{%
    $\begin{aligned} 
        &m_1 \begin{bmatrix}
            \dot{r}_1\\ \dot{v}_1
        \end{bmatrix}\!\!=\!\!\begin{bmatrix}
            m_1 v_1 \\ u_1-k_0 r_1 e^{-r_1}  - h_\nc{d}v_1-k_\nc{c} (r_1-r_2)
        \end{bmatrix}\!, \\
        &m_2 \begin{bmatrix}
            \dot{r}_2\\ \dot{v}_2
        \end{bmatrix}\!\!=\!\!\begin{bmatrix}
            m_2 v_2 \\ u_2\!-\!k_0 r_2 e^{-r_2}\!-\!h_\nc{d}v_2\!-\!k_\nc{c}(r_2\!-\!r_1)-k_\nc{c}(r_2\!-\!r_3)
        \end{bmatrix}\!, \ \ \\ 
        &m_3 \begin{bmatrix}
            \dot{r}_3\\ \dot{v}_3
        \end{bmatrix}\!\!=\!\!\begin{bmatrix}
            m_3 v_3 \\ u_3 -k_0 r_3 e^{-r_3}-h_\nc{d}v_3-k_\nc{c} (r_3-r_2)
        \end{bmatrix}\!, 
    \end{aligned}$ }
\end{equation}
\noindent where $r_i$, $v_i$ and $u_i$ denote respectively the position, velocity and input of subsystem $i\in \{1,2,3\}$. Likewise, {$k_0=1.1$ N/m}, $k_\nc{c}=0.25$ N/m, and $h_\nc{d}=0.30$~Ns/m are the springs stiffnesses and damping coefficients, and the masses are given by $m_1=1.5$ kg, $m_2=2$ kg, and $m_3=1$ kg. System~\eqref{eq:system_simulations} has been discretized using a sample time of $0.15$~s.  In addition, the controller's  objective function is defined by weighting matrices $Q=\text{diag}(2, \ 0.05, \ 2, \ 0.05, \ 2, \ 0.05)$ and $R= \text{diag}(0.1, \ 1, \ 0.1)$,  and the constraint sets are 
\begin{equation}
    \begin{split}
        &\mathcal{U}_i= \{u_i : -1.5 \leq u_i \leq 1.5\}, \\
        &\mathcal{X}_i= \left\lbrace\begin{bmatrix}
            r_i \\ v_i
        \end{bmatrix} : \begin{bmatrix}
            -5 \\ -2
        \end{bmatrix} \leq \begin{bmatrix}
            r_i \\ v_i
        \end{bmatrix} \leq \begin{bmatrix}
            5 \\ 2
        \end{bmatrix}\right\rbrace,
    \end{split}
\end{equation}
for $i \in \{1,2,3\}$. To design the terminal components, we have considered that they take the form $\kappa(x)=Kx$ and $V_\nc{f}(x) = x^\top P x$, being $K$ and $P$ matrices designed to satisfy~\eqref{eq:cond_term_comp}.\footnote{For the design of gain $K$ and  matrix~$P$ we have used Matlab\textsuperscript{\textregistered} \texttt{LMI Control Toolbox},  and for solving the MPC problems the solver \texttt{ipopt}. All simulations have been carried out in a~2.3~GHz 11th Gen Intel\textsuperscript{\textregistered}\! Core\textsuperscript{TM}\!  i7/16 GB RAM computer.}

We have first tested Algorithm~\ref{alg_1} with different combinations of $N_{\nc{c},2}$ and $N_{\nc{c},3}$, with~$N_{\nc{c},2}, N_{\nc{c},3}\in\{8,12,...,24\}$,  while fixing $N_{\nc{c},1}=10$ and defining $N_\nc{p}=\max_i N_{\nc{c},i}$. Subsequently, we have implemented the time-varying approach described in Section~\ref{sec:varying_Nc} using $\epsilon=5\cdot 10^{-6}$. Our simulations assess the different settings in terms of performance and computation time as illustrated below.

\begin{figure}[t]
    \centering
    \includegraphics[scale=0.52,trim={2.5cm 8.5cm 1cm 8.5cm},clip]{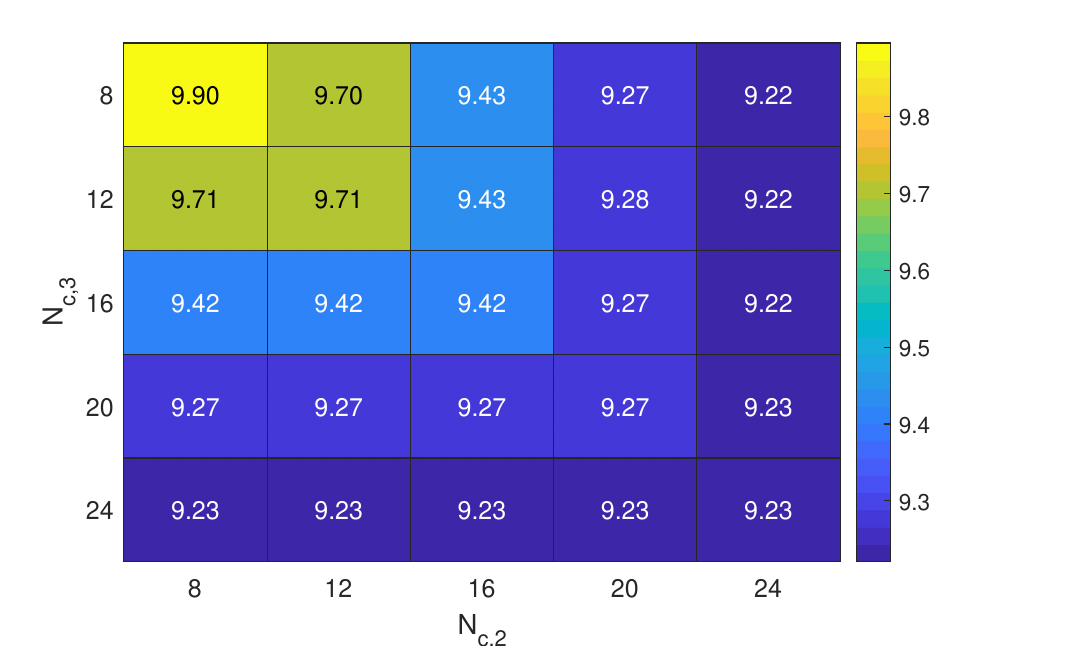}
    \vspace{-0.6cm}
    \caption{Heat-map with the values of index $J_\nc{cc}$ for different combinations of~$N_{\nc{c},2}$ and $N_{\nc{c},3}$.  }
    \label{fig:Jcc_3agents}
\end{figure}

\begin{figure}[t]
    \centering
    \includegraphics[scale=0.52,trim={2.5cm 9.5cm 1cm 9.5cm},clip]{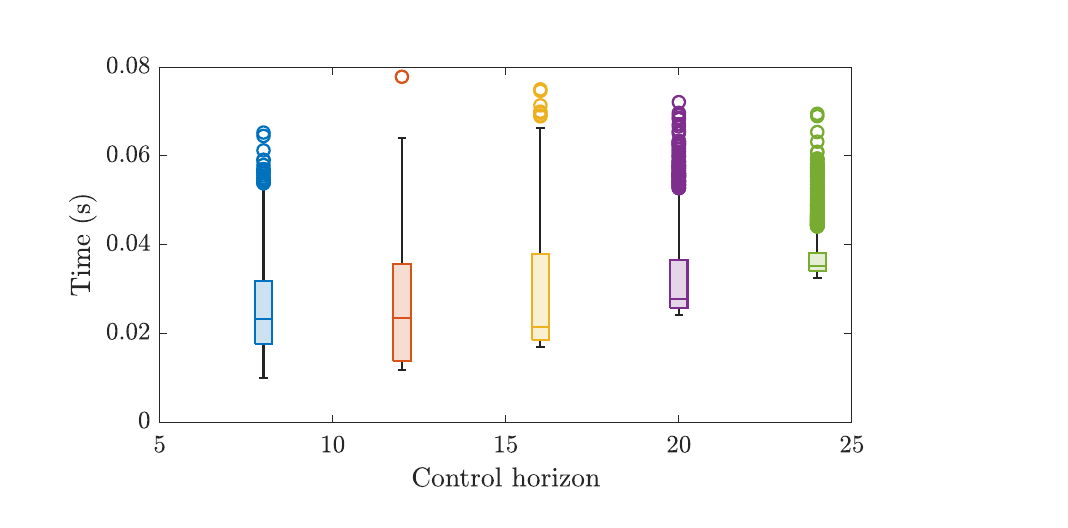}
    \vspace{-0.9cm}
    \caption{Box chart with the time required to solve problem~\eqref{eq:Dist_MPC} with different values of horizon~$N_{\nc{c},i}$. }
    \label{fig:CT_3agents}
\end{figure}

As for the system performance, Fig.~\ref{fig:Jcc_3agents} shows a heat-map with the values of index
\begin{equation}
    J_\nc{cc} = \sum_{k=0}^{T_\nc{sim}} \ell(x(k),u(k)),
\end{equation}
where $T_\nc{sim}$ denotes the number of simulated time steps, for the mentioned combinations of $N_{\nc{c},2}$ and $N_{\nc{c},3}$. As expected, increasing $N_{\nc{c},i}$ brings performance benefits. Minor discrepancies may also stem from both the convergence tolerance defined in Algorithm~\ref{alg_1} and the solver not providing the exact minimizers of the optimization problems. Nevertheless, it can be seen that there are combinations of $N_{\nc{c},2}$ and $N_{\nc{c},3}$ that provide superior performance than others while involving the same number of variables. For instance, if we set  $N_{\nc{c},2}=8$ and $N_{\nc{c},3}=16$, the cost decreases in 0.29 units in comparison with the case in which $N_{\nc{c},2}=N_{\nc{c},3}=12$. Similarly, $N_{\nc{c},2}=8$ and $N_{\nc{c},3}=20$ emulates the performance of $N_{\nc{c},2}=N_{\nc{c},3}=20$.  Likewise, it can be seen that greater relative benefits are obtained when increasing $N_{\nc{c},i}$ from 8 to~12, or 12 to 16, than when going from 20 to 24. 

As for the computation times, Fig.~\ref{fig:CT_3agents} provides a box chart with the time spent in solving nonlinear problem~\eqref{eq:Dist_MPC} with different values of $N_{\nc{c},i}$. This figure has been obtained considering all simulations and iterations performed. Additionally, as an example, Fig.~\ref{fig:xu_2agents} shows the state and input evolution over time for the case in which $N_{\nc{c},2}=8$ and $N_{\nc{c},3}=20$.

\begin{figure}[t]
    \centering
    \includegraphics[scale=0.52,trim={2.5cm 8.5cm 1cm 8.7cm},clip]{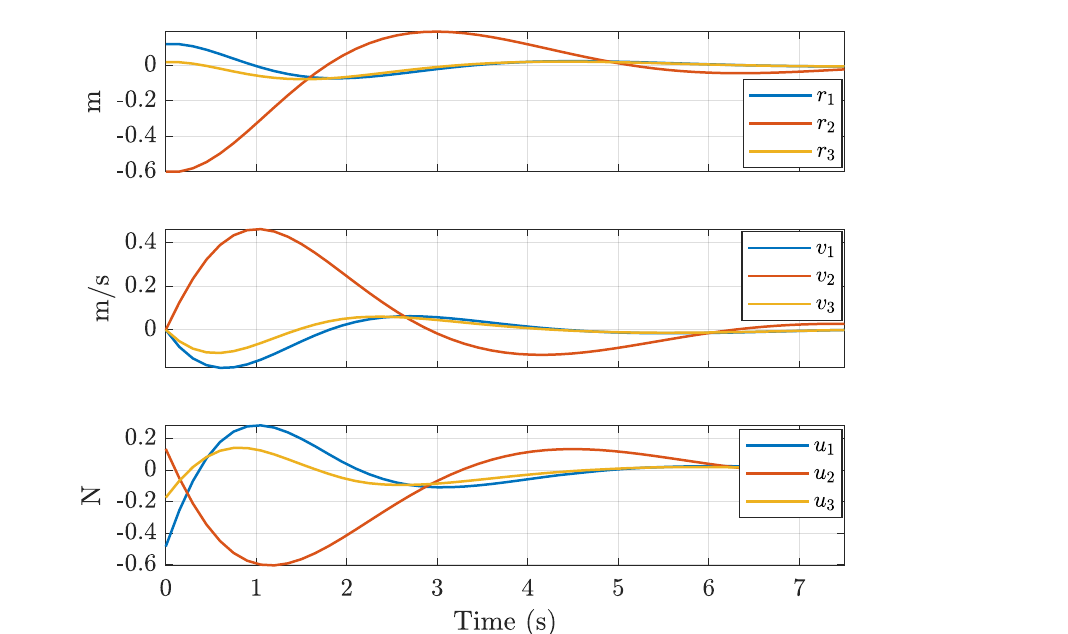}
    \vspace{-0.6cm}
    \caption{System state and input over time for $N_{\nc{c},2}=8$ and $N_{\nc{c},3}=20$.  }
    \label{fig:xu_2agents}
\end{figure}

\begin{figure}[t]
    \centering
    \includegraphics[scale=0.52,trim={2.5cm 10cm 1cm 10.7cm},clip]{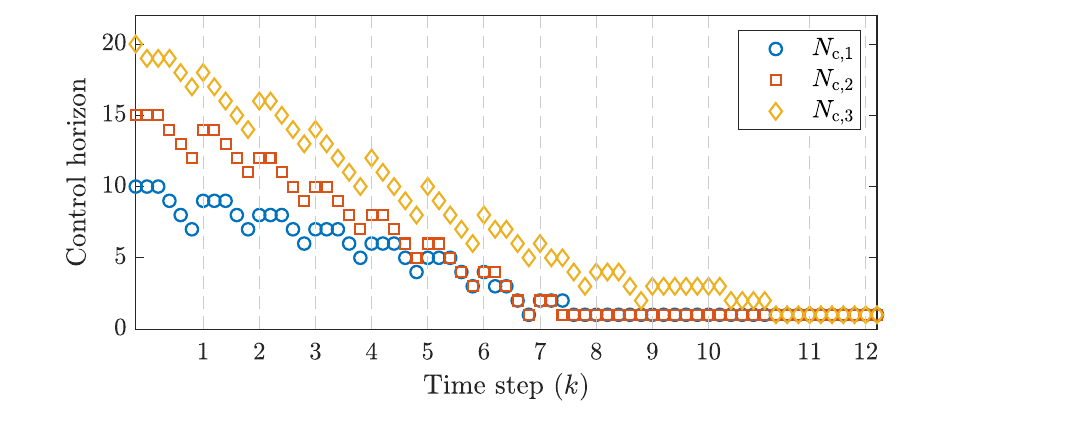}
    \vspace{-0.8cm}
    \caption{Evolution of $N_{\nc{c},i}$ for $i\in \{1,2,3\}$ over the iterations implemented in the first 12 time steps.  }
    \label{fig:Nc_var_masses}
\end{figure}


Finally, Fig.~\ref{fig:Nc_var_masses} illustrates the evolution of $N_{\nc{c},i}$ over the iterations performed at the initial 12 time steps when simulating the approach described in Section~\ref{sec:varying_Nc}. Specifically, at each time step and for $i\in \{1,2,3\}$, $N_{\nc{c}, i}^0$ is set to its mean during the iterations conducted at the preceding time step, starting from the values shown in Fig.~\ref{fig:Nc_var_masses} at time instant~0. As can be seen, the agents gradually reduce their control horizons and consequently the number of optimization variables. The grey vertical lines indicate time step changes, while each marker corresponds to an iteration for each of the agents. Note that the uneven spaces between these lines are due to the potential differences in the number of iterations performed at each time step. As a reference, the resulting performance cost in this simulation was 9.254.

\section{Conclusions}\label{sec:conclusions}

In this work, we have introduced a nonlinear cooperative DMPC approach to accommodate different control horizons for sets of agents. Our analysis and simulations, particularly focusing on a 3-agent system, have validated the feasibility of this approach and also highlighted its  potential. In particular, we have shown that dissimilar  control horizons can reduce the number of optimization variables while emulating the global performance of equal longer horizons. This can be beneficial for systems with limited computational resources and/or where the cost of computation is a critical factor.

Our simulations also point out that there is a phenomenon of diminishing returns as additional decision variables are negotiated. While this is not surprising because each extra variable corresponds to a \emph{farther} time within the prediction horizon, it suggests that beyond a certain point, the incremental benefits in performance gained from increasing coordination can be outweighed by the additional computational burden it brings. This can open up new possibilities for optimizing the DMPC framework, especially in terms of dynamically adjusting the level of coordination to achieve an optimal balance between performance and resources utilization. 

Our current efforts are indeed aligned with this idea, for we plan is to develop a DMPC framework that not only adapts to the varying capabilities and requirements of individual agents but also dynamically modulates the level of cooperation among them. 




\bibliography{IEEEexample.bib}


%








\bibliographystyle{IEEEtran}
\end{document}